\documentclass[12pt]{article}

\usepackage{a4wide}
\usepackage{amsmath}
\usepackage{amssymb}
\usepackage{amsthm}
\usepackage{subfigure}
\usepackage{url}
\usepackage{graphicx}
\usepackage{listings}

\newtheorem{theorem}{Theorem}
\newtheorem{lemma}{Lemma}
\newtheorem{definition}{Definition}
\newtheorem{corollary}{Corollary}

\newcommand{\comment}[1]{}

\newcommand{\scaff}[1]{#1^{\text{scaffold}}}
\renewcommand{\kill}[1]{#1^{\text{kill}}}
\newcommand{\lift}{\text{lift}}

\newcommand{\eps}{\varepsilon}

\newcommand{\co}{\colon\thinspace}

\newcommand{\R}{\mathbb{R}}

\newcommand{\Rd}{\mathbb{R}^d}
\newcommand{\N}{\mathbb{N}}

\newcommand{\order}{O}

\newcommand{\calR}{\mathcal R}
\newcommand{\calB}{\mathcal B}
\newcommand{\calH}{\mathcal H}

\newcommand{\pr}{\operatorname{pr}}

\title{Hardness of discrepancy computation and epsilon-net
  verification in high dimension}

\author{Panos Giannopoulos\thanks{Institut f\"ur Informatik, Freie Universit\"at Berlin, Takustrasse 9, D-14195 Berlin, Germany, \ \  
\texttt{panos@inf.fu-berlin.de}}
\and
Christian Knauer\thanks{Institut f\"ur Informatik,  Universit\"at Bayreuth, Universit{\"a}tsstrasse 30, 
95447 Bayreuth, Germany, \ \ \texttt{christian.knauer@uni-bayreuth.de}}{\ }\footnote{This research was supported by the German Science Foundation (DFG) under grant Kn~591/3-1.}
\and
Magnus Wahlstr{\"o}m\thanks{Max-Planck-Institut f{\"u}r Informatik, Stuhlsatzenhausweg 85,  66123 Saarbr{\"u}cken, Germany, \ \ \texttt{wahl@mpi-inf.mpg.de}}{\ }\footnote{This research was supported by the German Research Foundation (DFG) via its priority program "SPP 1307: Algorithm Engineering'', grant DO 749/4-1.}
\and
Daniel Werner\footnote{Institut f\"ur Informatik, Freie Universit\"at Berlin, Takustrasse 9, D-14195 Berlin, Germany, \ \ \texttt{daniel.werner@fu-berlin.de}}{\ }\footnote{This research was funded by Deutsche Forschungsgemeinschaft within the Research Training Group (Graduiertenkolleg) "Methods for Discrete Structures".}
}

\date{\today}                                      

\begin{document}

\maketitle

\begin{abstract}
Discrepancy measures how uniformly distributed a point set is with respect to a given set of ranges. There are two notions of discrepancy, namely \emph{continuous discrepancy} and \emph{combinatorial discrepancy}. Depending on the ranges, several possible variants arise, for example star discrepancy, box discrepancy, and discrepancy of half-spaces. In this paper, we investigate the hardness of these problems with respect to the dimension $d$ of the underlying space.

All these problems are solvable in time~$n^{\order(d)}$, but such a
time dependency quickly becomes intractable for high-dimensional data.
Thus it is interesting to ask whether the dependency on~$d$ can be
moderated.  

We answer this question negatively by proving that the canonical decision problems are W[1]-hard with respect to the dimension. This is done via a parameterized reduction from the \textsc{Clique} problem. As the parameter stays linear in the input parameter, the results moreover imply that these problems require $n^{\Omega(d)}$ time, unless \textsc{3-Sat} can be solved in $2^{o(n)}$ time. Further, we derive that testing whether a given set is an $\eps$-net with respect to half-spaces takes $n^{\Omega(d)}$ time under the same assumption. As intermediate results, we discover the W[1]-hardness of other well known problems, such as determining the largest empty star inside the unit cube. For this, we show that it is even hard to approximate within a factor of $2^n$.\\ \medskip


\end{abstract}

\section{Introduction}

Geometric discrepancy has significant applications in several areas, including optimization,
statistics, combinatorics, and computer graphics. See for example the textbooks by~\cite{Ma10},~\cite{Ch00}, and~\cite{DT97}. In particular, the \emph{star discrepancy} of a point set is important in
multi-variate numerical integration, 
where the error of a quasi-Monte Carlo integration is bounded as a 
function of the star discrepancy
of the point set used in the integration (by the Koksma--Hlawka 
inequality, see~\cite{niederreiterbook}). 
In addition, the difficulty of computing the star discrepancy can be 
an obstacle to evaluating different methods for creating 
low-discrepancy point sets, see, for example,~\cite{DBLP:journals/jc/DoerrGW10}.

Unfortunately, computing the star discrepancy of a point set using any 
known method
is computationally intensive; given a set~$P$ of~$n$ points in~$d$ 
dimensions,
every known method for getting even a constant-factor approximation of its star discrepancy 
has a running time of~$n^{\Theta(d)}$ (see below).
As many applications, for example in financial mathematics, require 
integration of functions
in tens or even hundreds of dimensions, this quickly becomes infeasible.

The main question we ask here is whether this dependency on~$d$ is
necessary.  Specifically, we ask whether the decision version for star discrepancy (and other related problems) can be
solved in ${O}(f(d)n^c)$ time, for some computable function $f$ and some
constant $c$ independent of $d$, i.e, whether it is fixed-parameter
tractable with respect to $d$. Note that $NP$-hardness
for a problem does not exclude such a possibility. Proving that a problem is
W[1]-hard with respect to $d$ implies that such an
algorithm is not possible, under standard complexity theoretic
assumptions. 

\subsection{Parameterized Complexity.}
We review some basic definitions of parameterized complexity theory;
see, for example, one of the textbooks by~\cite{DF99},~\cite{FG06}, and~\cite{Nie06} for an introduction.
A problem with input size $n$ and a positive integer parameter~$k$ is
\emph{fixed-parameter tractable} (fpt for short) if it can be solved by an
algorithm that runs in $\order(f(k)\cdot n^{c})$ time, where $f$ is a computable
function depending only on $k$, and $c$ is a constant independent of $k$;
such an algorithm is (informally) said to run in fpt-time.  The class of
all fixed-parameter tractable problems is denoted by FPT. 

An infinite hierarchy of classes, the W-hierarchy, has been introduced for establishing
fixed-parameter intractability.  Its first level, W[1], can be thought of
as the parameterized analog of NP: a parameterized problem that is hard for
W[1]\ is not in FPT\ unless FPT=W[1], which is considered highly
unlikely under standard complexity theoretic assumptions.  Hardness is
sought via an \emph{fpt-reduction}, i.\,e., an fpt-time many-one reduction
from a problem $\Pi$, parameterized with $k$, to a problem $\Pi'$,
parameterized with $k'$, such that $k'\leq g(k)$ for some computable
function $g$.

\subsection{Discrepancy and epsilon-nets}\label{Sec:Discrepancy}
In this section, we define the basic notion of discrepancy. Let $X$ be set and $\calR$ be a set of subsets of $X$, both not necessarily finite. A tuple $(X, \calR)$ is called a \emph{range space}.

If the range space arises from point sets and geometric objects, such as half-spaces or hyperrectangles (boxes), we call it a \emph{geometric range space}. Often, as in our case, the ranges are given implicitly. As an example, for a point set $P \subset \R^d$, we define 
\[ \calH_P := \left\{ H \cap P \mid H \text{ is a half-space} \right\}. \]
This is the range space induced by all half-spaces in $\R^d$. Observe that if $P$ is finite, even though there are infinitely many half-spaces, the size of $\calR$ is at most $2^{|P|}$.

We will now define two different notions of discrepancy, namely the \emph{continuous discrepancy} and the \emph{combinatorial discrepancy}.

\subsubsection{Continuous Discrepancy}
This concept relates the volume (i.e., its Lebesgue-measure) of a point set to its discrete measure, i.e., to the fraction of points it contains. To simplify matters, we restrict ourselves to point sets in the $d$-dimensional unit cube. 

The intuition is the following: a range space should have high discrepancy either if there is a range with small volume that contains a large fraction of points, or if there is a range with large volume that contains a small fraction of points. In that sense, it measures how good a finite set of points approximates the uniform distribution.

\begin{definition}\label{Definition:ContinuousDiscrepancy}
Let $P\subseteq \R^d$ be a set of points and $\calR$ be a set of subsets of $[0,1]^d$. We define the \emph{continuous discrepancy} of $P$ with respect to $\calR$ as
\[ \bar D_{\mathcal R}\left(P\right) := \max_{R \in \mathcal R} \left| vol(R) - \frac{|R \cap P|}{|P|} \right|. \]
\end{definition}
Figure \ref{fig:ContinuousDiscrepancy} shows two point sets in the plane with respect to axis-parallel boxes. The point set in Figure \ref{subfig:LowDiscrepancy} shows a point set with low discrepancy. In \ref{subfig:HighDiscrepancy1}, the discrepancy is attained for a box of high volume and few points inside. In \ref{subfig:HighDiscrepancy2}, it is attained for a small volume box with many points inside.
\begin{figure}
	\subfigure[low discrepancy set]{\includegraphics[width=0.32\textwidth]{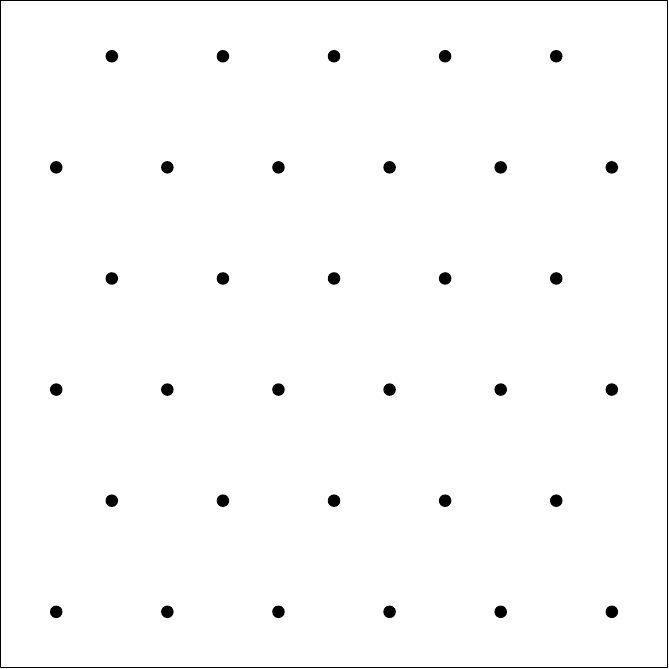}\label{subfig:LowDiscrepancy}}\hfill
	\subfigure[large box, few points]{\includegraphics[width=0.32\textwidth]{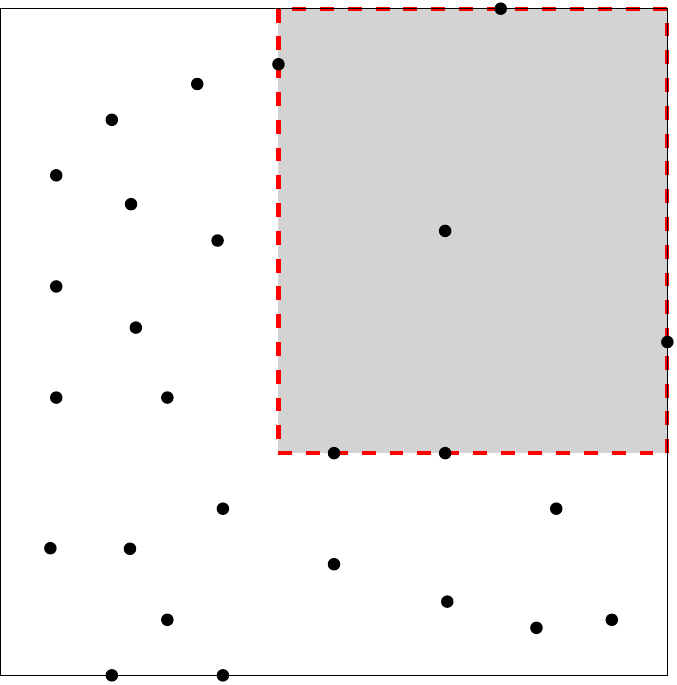}\label{subfig:HighDiscrepancy1}}\hfill
	\subfigure[small box, many points]{\includegraphics[width=0.32\textwidth]{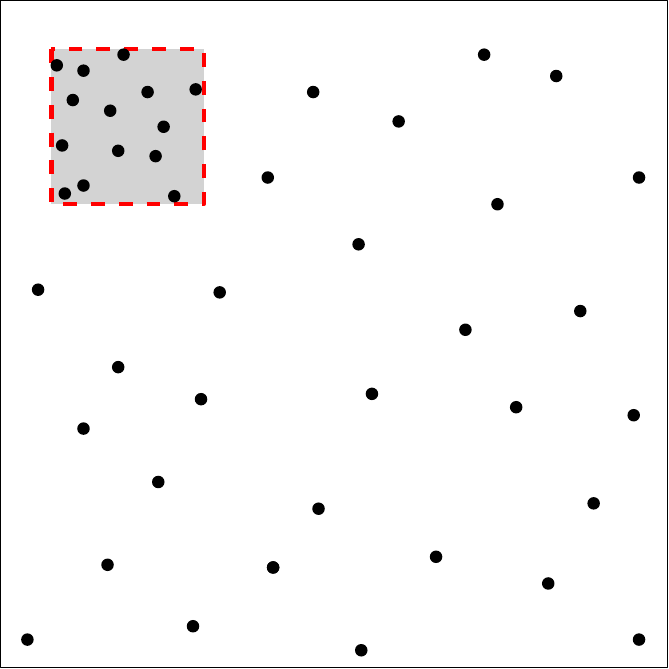}\label{subfig:HighDiscrepancy2}}
	\caption{Continuous Discrepancy}\label{fig:ContinuousDiscrepancy}
\end{figure}

\subsubsection{Combinatorial Discrepancy}
The \emph{combinatorial discrepancy}, sometimes called red-blue-discrepancy, is a slightly different notion. Here, we are given a set of points $P$, colored red or blue, and a set of ranges. Such a set is said to have high discrepancy, if there is a range where the difference between red and blue points is high.
\begin{definition}
Let $P = P_r \uplus P_b$ be a set of points in $\R^d$, and let $\calR$ be a set of subsets of $P$. We define the \emph{combinatorial discrepancy} of $P$ with respect to $\calR$ as
\[ \dot D_{\mathcal R}\left(P_r, P_b\right) := \max_{R \in \mathcal R} \left| | R \cap P_r | - | R \cap P_b | \right|. \]
\end{definition}

\subsection{Epsilon-nets}
A theory closely related to discrepancy is that of $\eps$-nets. We will give some basic terminology and afterwards discuss the relation of our results to $\eps$-net problems.

For a range space $(P, \calR)$, a set $S \subseteq P$ is called an $\eps$-net if for all $R \in \calR$ with $|R| \geq \eps |P|$ it holds that $R \cap S \neq \emptyset$. This means that $S$ intersects all large sets in $\calR$, namely those that contain at least an $\eps$-fraction of the points.

For applications, one is of course interested in nets that are small. For general range spaces, one can not expect such a behavior: If $\calR$ is the power-set of $P$, any $\eps$-net must be of size at least $n - \eps n + 1$.

Surprisingly, there are many range spaces where the size of a net does not depend on the value of $n$. 
The \emph{Vapnik-Chervonenkis dimension} has proved as a useful tool for studying these range spaces. Let $(P, \calR)$ be a range space. For a set $Q \subset P$, let $\pr_Q(P) := \left\{ R \cap Q \mid R \in \calR \right\}$ denote the projection onto $Q$. We say that a set $Q$ is \emph{shattered} by $\calR$, if $\pr_Q(P) = \mathcal P(Q)$, i.e., if all subsets of $Q$ appear in the projection of $P$ onto $Q$. Now the VC-dimension $\delta$ of a range space is the size of the largest set $Q \subset P$ that is shattered by $\calR$.

\cite{HW86} proved, based on a work by \cite{VC71}, that range spaces of finite VC-dimension $\delta$ admit $\eps$-nets of size $\order\left( \frac{\delta}{\eps} \log \frac{\delta}{\eps} \right)$. Their result even proves something a lot stronger, namely that a random sample of that size will be an $\eps$-net with high probability.

The theory of $\eps$-nets has been a hot topic recently, most notably because of the results by \cite{Al11} and \cite{PT10}.

\subsection{Our results}
We study the decision versions of several problems related to the above definitions.
In Section \ref{Sec:BichromaticRectangle}, we consider the following problem.
\begin{definition} ($d$-\textsc{Red-Blue-Discrepancy}) Let $P = P_r \uplus P_b $ be set of points in $\R^d$, and $k \in \N$. Decide whether $\dot D_\calB(P_r, P_b) \geq k$, where $\calB$ is the set of all axis-parallel boxes inside the unit cube.
\end{definition}

In particular, we show the following.
\begin{theorem}\label{Thm:RedBlueMain} The problem $d$-\textsc{Red-Blue-Discrepancy} is W[1]-hard with respect to the dimension and NP-hard.
\end{theorem}
This result will be easily derived by first showing the W[1]-hardness of another well known problem, called \textsc{Bichromatic-Rectangle}, where we have to find a box that contains as many blue points as possible, but no red points.

Subsequently, we will investigate the case of continuous discrepancy. Thereto, let $\calB_0$ be the set of axis-parallel boxes inside the unit cube containing the origin. Such a box is called a \emph{star}, or \emph{subinterval}. In Section \ref{Sec:StarDiscrepancy} and Section \ref{Sec:BoxDiscrepancy}, we will consider the following problem.
\begin{definition} ($d$-\textsc{Star-Discrepancy}) Let $P$ be a set of points in $\R^d$ and $V$ be a rational number. Decide whether $\bar D_{\calB_0}(P) \geq V$.
\end{definition}
The problem where the range space is the set of \emph{all} axis-parallel boxes inside the unit cube $\calB$ is defined analogously and will be called $d$-\textsc{Box-Discrepancy}. In Sections \ref{Sec:StarDiscrepancy} and \ref{Sec:BoxDiscrepancy}, we show the following. (The NP-hardness of star discrepancy was shown by~\cite{GSW09}.)
\begin{theorem}\label{Thm:ContinuousDiscrepancyMain} The problems $d$-\textsc{Star-Discrepancy} and $d$-\textsc{Box-Discrepancy} are W[1]-hard with respect to the dimension and NP-hard. 
\end{theorem}

In order to prove these two theorems, we consider two related problems, which have also been studied in the past.
\begin{definition} ($d$-\textsc{Maximum-Empty-Star}) Let $P$ be a set of points in $\R^d$ and $V$ be a rational number. Decide whether there is a box of volume $V$ inside the unit cube that contains the origin but no points from $P$.
\end{definition}
Analogously, we define the problem $d$-\textsc{Maximum-Empty-Box}. We establish the following results.
\begin{theorem}\label{Thm:EmptyBoxMain} The problems $d$-\textsc{Maximum-Empty-Star} and $d$-\textsc{Maximum-Empty-Box} are W[1]-hard with respect to the dimension. 
\end{theorem}

For the $d$-\textsc{Maximum-Empty-Star} problem we immediately get a result that is a lot stronger:
\begin{theorem}\label{Thm:Inapproximability} The problem $d$-\textsc{Maximum-Empty-Star} cannot be approximated in fpt time within a factor of $2^{|P|}$, unless $FPT = W[1]$.
\end{theorem}

Afterwards, we sketch hardness results for some other range spaces, such as boxes or half-spaces.

Finally, we build the connection to the theory of $\eps$-nets. It is known that for many sets, a small random sample will be an $\eps$-net with high probability. This leads to the following decision version.
\begin{definition} ($d$-\textsc{Net-Verification}) Let $P$ be a set of points in $\R^d$, $S$ be a subset of $P$, and an $\eps > 0$. Decide whether $S$ is an $\eps$-net for $P$ with respect to half-spaces.
\end{definition}
Our main theorem in this section shows that this question cannot be answered efficiently.
\begin{theorem}\label{Thm:EpsilonNetMain} The problem $d$-\textsc{Net-Verification} is co-NP-hard and co-W[1]-hard with respect to the dimension.
\end{theorem}

In all our reductions, the parameter $d$ is kept linear in the input parameter. Using a result by~\cite{CCFHJKX05}, we can derive something stronger.
\begin{corollary}\label{Cor:LowerBoundMain} All of the above problem cannot be solved in time $n^{o(d)}$, unless the Exponential Time Hypothesis fails, i.e., unless \textsc{3-Sat} can be solved in $2^{o(n)}$ time.
\end{corollary}

These results are obtained by fpt-reductions from
the W[1]-complete $k$-\textsc{Clique} problem in general graphs, see~\cite{DF99}, based
on the general framework by~\cite{cgkmr-gcfpt-09, CGKR08}.

\subsection{Related work.}
When the dimension is part of the input, the \textsc{Bichromatic-Rectangle} problem was shown to be $NP$-hard by~\cite{606903}; in the same paper an $\order(n^{2d+1})$-time algorithm
was given. \cite{BK09} gave an algorithm that runs in
$\order(k\log^{d-2}n)$ time, where $k$ is the number of feasible boxes
that are not properly contained in any feasible box, and showed
that $k$ can be $\Theta(n^d)$ in the worst case. \cite{AH08} gave an $(1-\eps)$-approximation algorithm that
runs in $\order(n^{\lceil d/2\rceil}(\eps^{-2}\log n)^{\lceil d/2 +1\rceil})$ time.

The \textsc{Star-Discrepancy} problem has been shown to be $NP$-hard by~\cite{GSW09}.  An exact algorithm that runs in $\order(n^{1+d/2})$ time was given
by~\cite{DEM96}. \cite{DBLP:journals/jc/Thiemard01} has given
an approximation algorithm that achieves additive error and runs in
fpt-time with respect to the error and the dimension. However, as~\cite{DBLP:journals/jc/Gnewuch08} noted, by setting the error
tolerance to the same order as the discrepancy of an optimal point set, so
that a constant-factor approximation is achieved, the running time of any
algorithm following Thi\'emard's approach becomes $n^{O(d)}$. 
As for the \textsc{Box-Discrepancy}, no hardness results where known so far.

The \textsc{Maximum-Empty-Box} problem has been studied extensively in the
planar case, see for example~\cite{AS87} and references
therein.  When the dimension is part of the input, the problem has only
recently been shown to be $NP$-hard by~\cite{BK10} and the fastest exact
algorithm runs in time $\order(n^d\log^{d-2}n)$~\cite{BK10}. 
Also recently,~\cite{DJ09} gave an $\order((8ed\eps^{-2})^d\cdot 
n\log^{d}n)$-time $(1-\eps)$-approximation algorithm for this problem.
Note that, since $(\log n)^k < n + f(k)$ for some $f(k)$, this counts as fpt time 
in parameters $1/\epsilon$ and $d$, in contrast to our results for \textsc{Maximum-Empty-Star}.
The NP-hardness of the \textsc{Maximum-Empty-Star} problem was shown by~\cite{GSW09}.

\section{Red-Blue Discrepancy and the Bichromatic Rectangle Problem}\label{Sec:BichromaticRectangle}
In order to show the hardness of \textsc{Red-Blue-Discrepancy}, we will first consider the \textsc{Bichromatic-Rectangle} problem. The parameterized decision problem is defined as follows:
\begin{definition} ($k$-$d$-\textsc{Bichromatic-Rectangle})
Let $P_r$ be a set of red points and a set $P_b$ be a set of blue points in $\Rd$, and $k \in \N$. Decide whether there is an axis-parallel box such that $H \cap P_r = \emptyset$ and $|H \cap P_b| \geq k$? 
\end{definition}

A box that does not contain any point from $P_r$ will be called \emph{feasible}. For a given set of points $P = P_r \uplus P_b$, let 
\[ E_{\calB}(P_r, P_b) = \max_{B \in \calB, B \cap P_r = \emptyset} |B \cap P_b|\]
denote the size of an optimal solution. Recall that $\calB$ is the set of all axis-parallel boxes inside the unit cube.

\subsection{The idea.}
In order to show that the $k$-$d$-\textsc{Bichromatic-Rectangle} problem
is $W[1]$-hard, we will give a reduction from the $k$-\textsc{Clique} problem. For a given simple graph $G = ([n], E)$ we will construct sets
$P_r = P_r(G, k)$ and $P_b = P_b(G, k)$ in $\R^{2k}$ such that $G$ has a clique of size
$k$ if and only if $E(P_r, P_b) = k+1$.

In addition to the (blue) origin $0$, we will put blue and red points into
$k$ pairwise orthogonal two-dimensional planes. These points will be used to
encode the vertices of $G$. Additional red points will then be used to encode the
edge-set of $G$.

Each plane will contain $n$ blue points, corresponding to the vertices of
the graph, and $n-1$ red points. The red points are placed such that no feasible box can contain more than one blue
point from a single plane. Thus, at most $k$ of these blue points can be
contained in any feasible box. We will then ensure that
such a box can only contain points $x$ and $y$ from two
different planes if the corresponding vertices are connected in $G$. This
is done by putting red points into the product of the respective
planes (which is a four-dimensional subspace). 

This construction will ensure that any feasible box containing $k+1$ blue points corresponds to a 
$k$-clique in $G$, and vice versa.

\subsection{Preparations.}
For $1 \leq i \leq k$, we define the two-dimensional subspace
\[ \R_i^2 = \left\{ (x_1, y_1, \dots, x_k, y_k) \mid x_j = y_j = 0, j \neq
  i) \right\} \subseteq \R^{2k}. \]
For $1 \leq i < j \leq k$, we set $\R_{ij}^4$ to be the product of $\R_i^2$ and $\R_j^2$, i.e., $\R_{ij}^4 = \R_i^2 \times \R_j^2$.\\
For $p \in \R_i^2$ and $q \in \R_j^2$, observe that the unique
point in $\R_{ij}^4$ that (orthogonally) projects to $p$ (into $\R_i^2$)
and to $q$ (into $\R_j^2$) is $p + q$.

\subsection{The scaffold construction.}
Let $\eps = 1/4$. For a vertex $1 \leq v \leq n$, we define the point
 \[ b_i(v) = (v, n + 1 - v) \in \R_i^2. \] 
Then, let
\[ (P_b)_i^{\text{scaffold}} = \{ b_i(1), \dots, b_i(n) \} \subseteq \R_i^2 \]
be the set of all points in the $i$-th plane.
Choosing a (rectangle containing) point $b_i(v)$ will correspond to
choosing vertex $v$ from $G$. 
Let \[ \scaff P_b = \biguplus_{1 \leq i \leq k} \scaff {(P_b)}_i \] be the set of all these blue points.

As we want the feasible boxes to contain at most one point from
each $\R_i^2$, we add a set of red points as follows: For $1 \leq v \leq
n-1$, we define $r_i(v) = (v + 1/2, n + 1 - (v + 1/2))$ and set
\[ (P_r)_i^{\text{scaffold}} = \{ r_i(1), \dots, r_i(n - 1) \} \subseteq \R_i^2. \]
Finally, we define
\[ \scaff P_r = \biguplus_{1 \leq i \leq k} \scaff {(P_r)}_i \]
to be the set of all red scaffolding points. See Figure \ref{fig:ScaffoldSelection} for an example of the scaffold construction.
\begin{figure}
	\centering
		\includegraphics[width=0.50\textwidth]{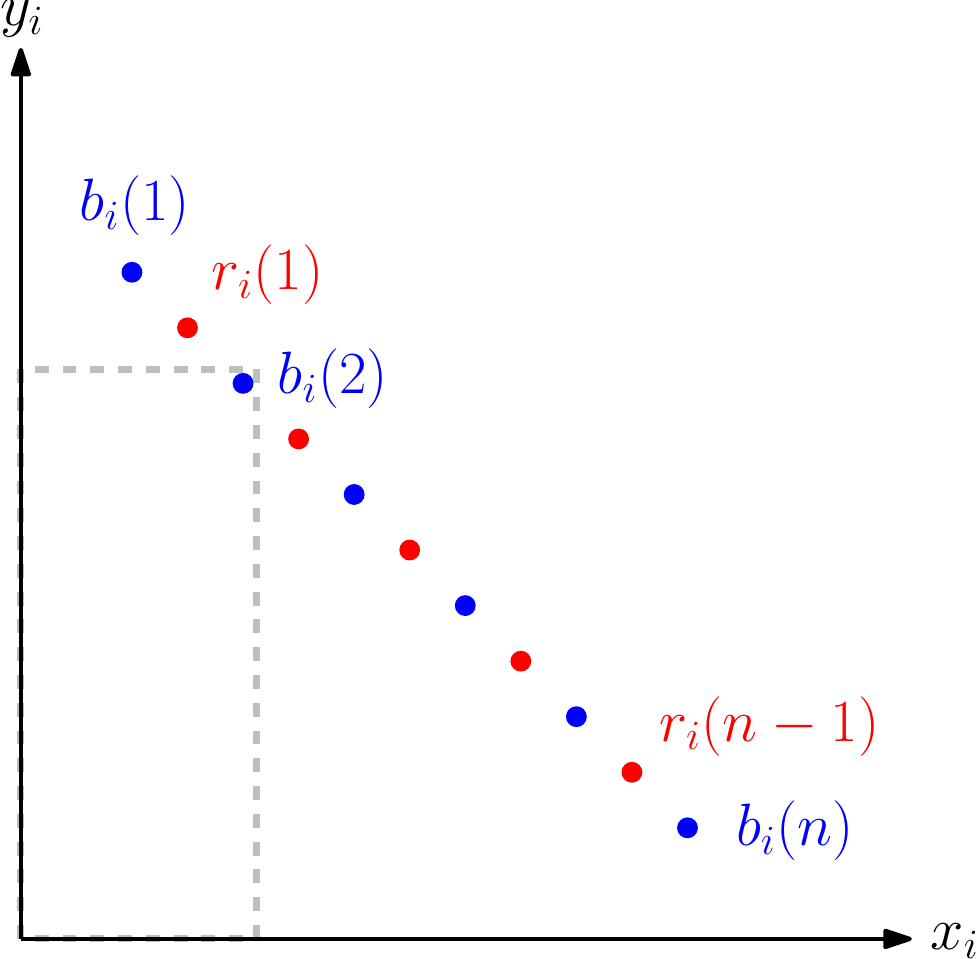}
	\caption{The scaffold construction with vertex $2$ selected.}
	\label{fig:ScaffoldSelection}
\end{figure}
Observe that an feasible box $B$ can contain at most one blue point from each $\scaff{(P_b)}_i$.

\subsection{Encoding edges.}
In order to encode the edges of the graph, we will place several red points
between pairs of $\R_i^2$s. This will forbid certain pairs of blue points to be
selected at the same time, namely the ones that correspond to vertices not being connected in $G$.

Thereto, for $1 \leq i \leq j \leq k$
and vertices $1 \leq u, v \leq n$, we define the point
\[ r_{ij}^{\text{kill}}(uv) = b_i(u) + b_j(v) \in \R_{ij}^4.\]
The crucial property of such a point is that it is inside a box (containing the origin) if and only if both $b_i(u)$ and $b_j(v)$ are inside this box.

The set of all killing points in $\R_{ij}^4$ is then
\[(P_r)_{ij}^{\text E} = \{ r_{ij}^{\text{kill}}(uv), r_{ij}^{\text{kill}}(vu)
\mid uv \notin E \}.\] 
As the graph is simple (i.e., contains no loops), so all
points of the form $r_{ij}(uu)$ are also added. Finally, we set
\[ P_r^{\text E} = \biguplus_{1 \leq i \neq j \leq k} {(P_r)}_{ij}^{\text E}\]
to be the set of all killing points. See Figure \ref{fig:KillingPoints} for an example where $uv \notin E$.
\begin{figure}[htbp]
	\centering
		\includegraphics[width=0.7\textwidth]{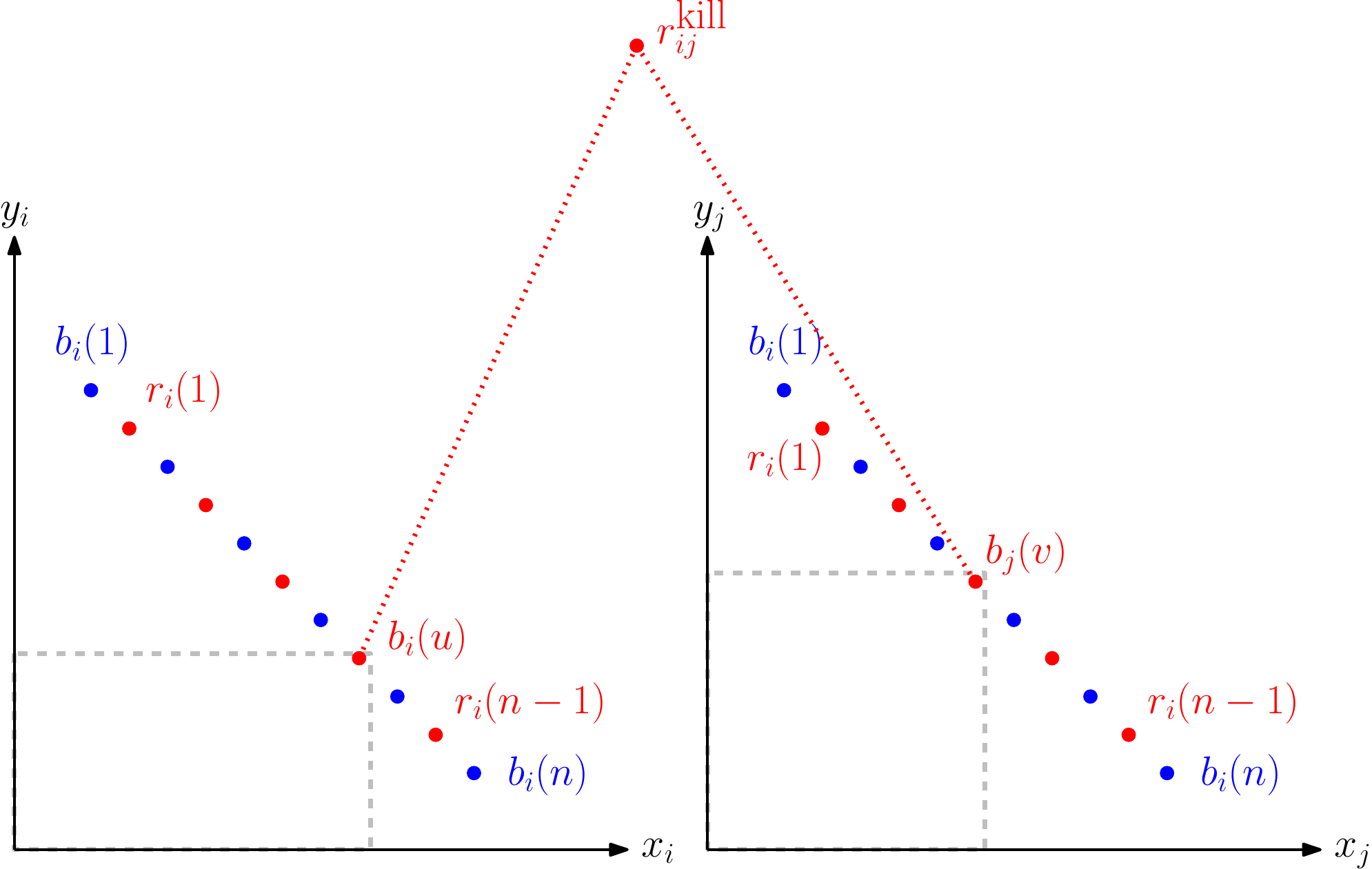}
	\caption{$b_i(u)$ is the projection of $\kill r_{ij}(uv)$ to $\R_i^2$ and $b_j(v)$ is the projection of $\kill r_{ij}(uv)$ to $\R_j^2$.}
	\label{fig:KillingPoints}
\end{figure}

\subsection{The overall construction.}
For $G = ([n], E)$ and $k > 0$ we construct point sets $P_r(G, k), P_b(G, k)$
in $\R^{2k}$ as follows:
\begin{itemize}
\item $P_r(G, k) = P_r^{\text{scaffold}} \cup P_r^{\text E}$
\item $P_b(G, k) = \{ 0 \} \cup \scaff P_b$
\end{itemize}
The size of the point set is $\order(k^2n^2)$ and the coordinates of the
points can be encoded by $\order(\log kn)$ many bits.  Clearly the construction
can be performed in time polynomial in both $k$ and $n$.

\begin{lemma}\label{Lemma:RedBlueMain} $G$ has a $k$-clique if and only if $E_{\calB}(P_r, P_b) = k+1$.
\end{lemma}
\begin{proof} First, observe that any feasible box $B$ can contain at most $k+1$ points, as $|B \cap \scaff {(P_b)}_i| \leq 1$, for $1 \leq i \leq k$. Additionally, $0$ can be in $B$. 

Let $v_1, \dots , v_k$ be a clique of size $k$. We choose a (closed) box $B$ with upper right corner $b_i (v_i)$ in $\R^2_i$. 

$B$ contains exactly one point from each of the $\R_i^2$, and also the origin, making it a total of $k+1$ points. We show that $B$ is feasible. First, by definition $B$ contains no point of $\scaff{P_r}$. Further, assume that $B$ contains a point of $P_r^E$, say $\kill{r}_{ij}(uv) = b_i(u) + b_j(v) \in {(P_r)}_{ij}^E$. Then $B$ contains both $b_i(u)$ and $b_j(v)$. But this means $uv \notin E$, for otherwise the point $\kill{r}_{ij}(uv)$ would not have been added; a contradiction.

Now assume that there is no clique of size $k$. Let $B$ be any
box containing $k+1$ points. We show that $B$ is infeasible. If
$B$ contains a red point from one of the $\R_i^2$'s, we are
done. Otherwise, besides the origin it can contain at most one blue point from each
$\R_i^2$. Let $u = v_i$ and $v = v_j$ be two vertices corresponding to blue
points contained in $B$ that are not connected in $G$. As there is no
$k$-clique, such a pair must exist. Then $B$ also contains the red point $\kill{r}_{ij}(uv)$. Thus,
$B$ is infeasible. \qed
\end{proof}
\begin{theorem} The $k$-$d$-\textsc{Bichromatic-Rectangle} problem is $W[1]$-hard when parameterized with both the dimension $d$ and the size of the solution $k$.
\end{theorem}

As noted by~\cite{DBLP:journals/ipl/CesatiT97} an optimization problem that is $W[1]$-hard when parameterized by the size of the solution is unlikely to have 
an approximation scheme that runs in $f(\eps)n^c$, i.e, an efficient polynomial-time approximation scheme (EPTAS). In our case, since the problem is
hard with respect to both the dimension and the size of the solution, this implies the following:

\begin{corollary}
  The (optimization version of the)
  $k$-$d$-\textsc{Bichromatic-Rectangle} problem does not admit an
  approximation scheme that runs in $\order(f(1/\eps, d)\cdot
  \textup{poly}(n,d))$ time, unless $W[1]=FPT$.
\end{corollary}

\subsection{Adaption to Red-Blue Discrepancy}
In order to adapt this proof to the \textsc{Red-Blue-Discrepancy} problem, we have to modify the set in such a way that a clique corresponds to a set with high discrepancy. Thereto, let $N$ be the total number of points in the construction. We replace the blue point at the origin by $N$ copies. Now the value
\[ \dot D_{\mathcal R}\left(P_r, P_b\right) := \max_{R \in \mathcal R} \left| | R \cap P_r | - | R \cap P_b | \right| \]
is maximized for a box with many blue points inside (as it is at least $N$). Observing that, for each $\R^2_i$, the difference between blue and red points is at most one, we can follow the above reasoning. This means there is a box with $N + k$ blue points (and no red points) inside, if and only if $G$ has a $k$-clique. This proves Theorem \ref{Thm:RedBlueMain}.

\section{The Maximum Empty Star Problem}\label{Sec:MaximumEmptyStar}
We now turn to the continuous version of the problems. In this section, we consider the problem where we have to find an empty axis-parallel box inside the unit-cube of maximal volume that contains the origin, namely \textsc{Maximum-Empty-Star}.
Besides showing the W[1]-hardness, our construction yields that it is even $W[1]$-hard to approximate this problem by a factor of $\left(\frac{1}{2}\right)^{|R|}$.

The proof uses ideas similar to the discrete case in Section \ref{Sec:BichromaticRectangle}. As above, because the origin has to be included in the boxes, the planes will be considered separately. In this construction, the analog of a rectangle containing a blue point from one of the $\R_i^2$ is now a rectangle that is "large" (of size $C$ for some $0 < C < 1$ to be determined later). In each plane, there will be $n$ large rectangles to choose from, corresponding to the $n$ vertices of $G$. It will only be possible to choose large rectangles from two different planes, if the corresponding vertices are connected in $G$. This yields a one-to-one correspondence between "large" empty boxes and cliques of size $k$.

\subsection{The Construction.}

We will proceed as follows: First, we determine where the upper right
corners of the $n$ large rectangles have to be. From this, we will
determine the blocking points (which are the analog of the red points
above) that are needed for this.

Let $\mu > 1$ be a parameter to be specified later. One possibility to
determine the upper right corners of the rectangles, each having area $C =
\frac{1}{\mu^{n-1}}$ in one $\R_i^2$, is as follows:
\[ c_i(u) = \left(C\mu^{u-1}, \frac{1}{\mu^{u-1}}\right), 1 \leq u \leq
n.\] We now place points such that any maximal empty (open)
rectangle, i.e., a rectangle supported by two points, has its upper right
corner at one $c_i(u)$.  This can be realized by the following
blocking points:
\[ p_i(u) = \left(C\mu^{u-1}, \frac{1}{\mu^u}\right), 0 \leq u \leq n. \]
We set $\scaff P_i = \{ p_i(u) \mid 0 \leq u \leq n\}$ and $\scaff P =
\uplus_{1 \leq i \leq k} \scaff P_i$. 

Thus, in each $\R_i^2$, we have $n$ choices for the upper right corner of the rectangles: the points $c_i(u)$, $1 \leq u \leq n$.
If a rectangle has its upper right point somewhere
else on $(x, C/x)$ or above, it contains a point from $\scaff P$, and any other
feasible rectangle has smaller size. 

Choosing a large rectangle in
each of the $\R_i^2$ gives us an empty rectangle of total volume $C^k$. See
Figure \ref{fig:ChristiansConstruction} for an example.

\begin{figure}[htbp]
	\centering
		\includegraphics[width=0.60\textwidth]{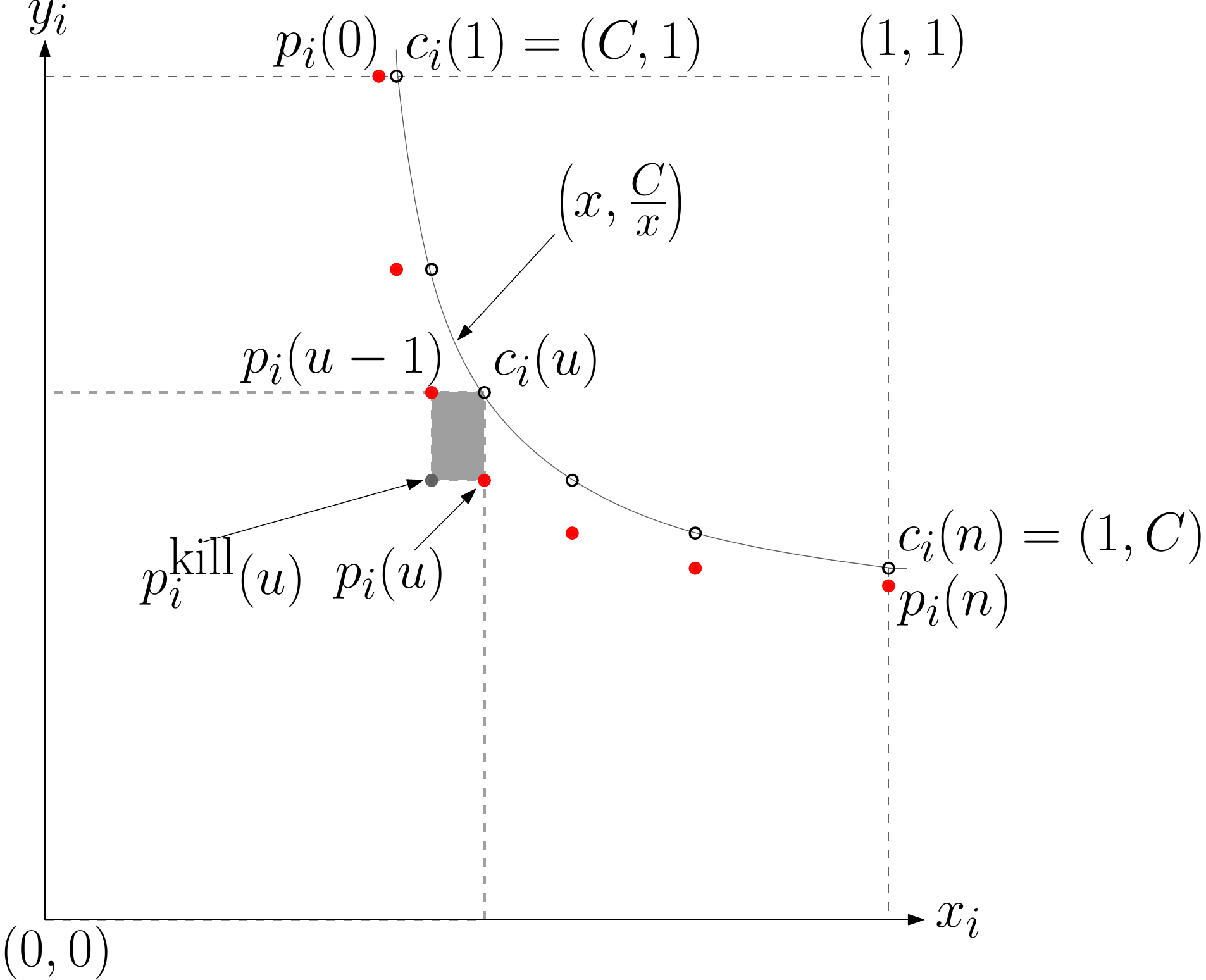}
	\caption{The plane $\R_i^2$. A rectangle selecting vertex $u$ is indicated.}
	\label{fig:ChristiansConstruction}
\end{figure}

\subsection{Encoding the edges.}
As above, if the vertices corresponding to two different large rectangles in
the planes $\R_i^2$ and $\R_j^2$ are not connected, we will add a point in
the product $\R_{ij}^4$ that forbids these two rectangles to be chosen at
the same time. To this end, we let 
\[ \kill p_i(u) = (C \mu^{u-2},
1/\mu^{u}) \] 
(these points are themselves not added to the set $P$)
and then define $\kill p_{ij}(uv) = \kill p_i(u) + \kill p_j(v)$.
Recall that this also includes all points of the form $p_{ij}(uu)$. Then we set
\[ P^{\textnormal E} = \{\kill p_{ij}(uv) \mid i \neq j, uv \notin E  \}, \]
and
\[ P = P^{\textnormal E} \cup \scaff P. \] 
The size of $P$ is $\order(n^2
k^2)$. If we set $\mu=2$, all coordinates have size polynomial in the size
of the input. We also let $V=C^k=1/\mu^{k(n-1)}$. Clearly the construction can be
performed in time polynomial in $k$ and $n$.

Now we come to prove the correctness of the construction. Let $F_i(u)$ be the rectangle with corners $p_i(u-1)$, $c_i(u)$, $p_i(u)$, $\kill p_i(u)$, as indicated in Figure \ref{fig:ChristiansConstruction}.

\begin{lemma}\label{Lemma:RegionLemma} Any feasible rectangle in $\R_i^2$ that does not intersect any region $F_i(u)$, $1 \leq u \leq n$, has size at most $C/\mu$.
\end{lemma}
\begin{proof} Such a rectangle has its upper right point below the graph going through the points $p_i(u)$, $1 \leq u \leq n$, which is $(x, \frac{C}{x\mu})$. \qed
\end{proof}
We use this to prove the main Lemma, the continuous analog of Lemma~\ref{Lemma:RedBlueMain}:
\begin{lemma}\label{Lemma:MaxEmptyStar} $G$ has a $k$-clique if and only if there is an empty
  star of size $V=C^k$. Further, if $G$ does not have a $k$-clique, the largest empty star has volume $C^k/\mu$.
\end{lemma}
\begin{proof} Let $v_1, \dots v_k$ be a clique in $G$. In each $\R_i^2$, $1
  \leq i \leq k$, choose the rectangle with upper right corner $(C\cdot
  \mu^{v_i - 1}, \frac{1}{\mu^{v_i - 1}})$. Then the box is the product of these
  rectangles and it has volume $C^k$. By definition, it does not contain
  any point from one of the $\R_i^2$. If it would contain a point
  $p_{ij}(uv) \in \R_{ij}$, then the projection of $p_{ij}(uv)$ onto both
  $\R_i^2$ and $\R_j^2$ would be contained in the corresponding
  rectangles. But this means that $uv \notin E$, a contradiction.

  If there is no $k$-clique, any selection of $k$ large rectangles that do
  not contain a point in $\R_i^2$ would contain a point in one of the
  $\R_{ij}^4$, as among any $k$ vertices there are at least two that are not connected. Thus, in order to avoid all points, by Lemma
  \ref{Lemma:RegionLemma} in at least one $\R_i^2$ we cannot select a large
  rectangle intersecting any of the $F_i(u)$. Thus, the total volume can be
  at most $C^{k-1}\cdot \frac{C}{\mu} = C^k/\mu$.  \qed
\end{proof}
Thus, we have shown the first part of Theorem \ref{Thm:EmptyBoxMain}.

\subsection{An inapproximability result.}
Now it easily follows that the problem is even hard to approximate:
the $\mu$ chosen above can be made very large. By Lemma
\ref{Lemma:MaxEmptyStar}, the ratio between a large box for a point set
constructed from a positive instance and one constructed from a negative
instance is at least $\mu$. Since we can choose $\mu = 2^{\order(|R|)}$
(this takes only polynomially many bits), Theorem \ref{Thm:Inapproximability} follows.


\section{The Star Discrepancy Problem}\label{Sec:StarDiscrepancy}

In this section we show that computing the star discrepancy of a point set
inside the unit cube is $W[1]$-hard.

There are two reasons why the previous reduction does not give us the hardness-result for this problem right away: 
\begin{itemize}
	\item 	First, the maximum discrepancy can be
			attained by either a large box with few points inside or by \emph{a small box
			with many points inside}. For example, in our construction from Section \ref{Sec:MaximumEmptyStar}, large point 				sets lie in a box with (affine) dimension $d-1$ and thus a volume of 0. 

	\item 	Second, even if the maximum is attained for a large box, it might still contain \emph{some}
			points, in which case our construction would fail. 
\end{itemize}
However, we can get rid of
both problems by simply choosing the right value for $\mu$, and thus, $C$. Recall that these values determined the size of the largest empty box. $C$ is exactly the area of a maximum empty rectangle in each $\R_i^2$.

For a graph $G$, let $N$ be the total number of points in our construction
from the previous section. Recall that $N \in \order(k^2n^2)$. Any
box (containing the origin) that contains all points has volume of $1$, as there are points
with $x_i = 1$ and $y_i = 1$ for all $1 \leq i \leq k$. This leads to the following observation.
\begin{corollary}\label{Cor:MaxPointDis} For all boxes $B$ containing the origin we have 
\[ \frac{|B \cap P|}{|P|} - vol(B) \leq \frac{N-1}{N}. \]
That means that the fraction of points in any box can be bigger by at most $(N-1)/N$ than its volume.
\end{corollary}

Thus, our construction from Section \ref{Sec:MaximumEmptyStar} works if we can ensure that the largest empty box in a positive instance has volume more than $(N-1)/N$, that means 
\[ C^k = \left( \frac{1}{\mu^{n-1}} \right)^k > \left( \frac{N-1}{N} \right). \] 
Then the discrepancy is attained for an empty box: Enlarging the box can increase the volume by at most $1 - C^k$. But as $C^k > \frac{N-1}{N} = (1 - \frac{1}{N})$, we have that $(1 - C^k) < 1/N$. Thus, picking such an extra point cannot increase the discrepancy. This means that we can choose $\mu$ such that
\[ 1 < \mu < \left( \frac{N}{N-1} \right)^{\frac{1}{k(n-1)}}. \]

To make sure that $\mu$ requires only polynomially (in
$k$ and $n$) many bits, observe the following.
\begin{lemma}\label{Lemma:Mu} For $\mu = 1 + \frac{1}{t}$ with $t = 2knN$, it holds that $\mu^{k(n-1)} < \frac{N}{N-1}$.
\end{lemma}
\begin{proof}
Observe that $\frac{N}{N-1} = \sum_{i = 0}^{\infty}  \left( \frac{1}{N} \right)^i$. 
Then 
\[\mu^{k(n-1)} < \mu^{kn} = \left(1 + \frac{1}{t}\right)^{kn} = \sum_{i = 0}^{kn}{kn \choose i} \frac{1}{t^i} \leq \sum_{i = 0}^{kn}(kn)^i \frac{1}{t^i} \leq \sum_{i = 0}^{\infty} \left( \frac{kn}{t} \right)^i. \]
Thus,
\[\mu^{k(n-1)} = \left(1 + \frac{1}{t}\right)^{k(n-1)} < \sum_{i=0}^{\infty}\left( \frac{1}{2N} \right)^i < \sum_{i = 0}^{\infty} \left( \frac{1}{N} \right)^i = \frac{N}{N-1}. \] \qed
\end{proof}

Constructing the set $P$ with this value of $\mu$, we immediately get:
\begin{lemma}\label{Lemma:StarDisrepancyMain} $G$ has a clique of size $k$, if and only if $\bar D_{\calB_0}(P) = C^k$.
\end{lemma}
\begin{proof} By the previous remarks, the maximum is attained for a \emph{large empty} box. Then, the proof follows from Lemma \ref{Lemma:MaxEmptyStar}. \qed
\end{proof}

This proves the first part of Theorem \ref{Thm:ContinuousDiscrepancyMain}.



\section{Largest Empty Box Problem}\label{Sec:LargestEmptyBox}
In the two upcoming sections, we will consider the analogous problems for the case when the origin does not have to be contained in the boxes. That means our range space is $\calB$, the set of all axis-parallel boxes inside the unit cube, as defined in Section \ref{Sec:Discrepancy}. We start with the case where we have to find a large empty box inside the unit cube, namely $d$-\textsc{Maximum-Empty-Box}.

This problem is quite different in the sense that, as the box does not have to contain the origin, now the $\R_i^2$ cannot be considered separately any more. This kills our construction from the previous section: the box $(0, 1)^{2k}$ does not contain any points from $P$ but has volume $1$. 

The plan is to reestablish this dependence, so that we can use the same reasoning as above. This can be done by a simple trick, which we call \emph{lifting}: From a graph $G$, we first construct the set $P$ as in Section \ref{Sec:MaximumEmptyStar} with the constant $C^k = 2/3$. Then, define the function $\lift\co \R^{2k} \to \R^{2k}$ as follows: 
\[ \lift (x_1, \dots, x_{2k}) = (x_1', \dots, x_{2k}')  \text { with } x_i' = \begin{cases}  x_i  & \text{ if } x_i \neq 0 \\  x_i' = 1/2 & \text{ otherwise}. \end{cases}\]
Now we apply the function lift to all points in the set $P$. For the lifted point $x$, we call the $\R_i^2$ that the point was lifted from the \emph{corresponding} $\R_i^2$. This gives the following:
\begin{lemma}\label{Lemma:LiftingLemma} Any box $B \in \calB$ having volume at least $2/3$ contains a point $x$ if and only if the projection onto the corresponding $\R_i^2$ contains the projection of $x$.
\end{lemma}
\begin{proof} As the box has volume at least $2/3$, each of its projections onto any of the $\R_j^2$ has an area of at least $2/3$, for otherwise the total volume would be less than $2/3 \cdot 1$. Thus, $(1/2, 1/2)$ is contained in the projection of $b$ to $\R_j^2$, for all $1 \leq j \leq k$. 

This means that every large box (of volume at least $2/3$) contains every point in all dimensions, except possibly for the corresponding $\R_i^2$. From this, the claim follows. \qed
\end{proof}

Further, any box of volume $2/3$ has its lower left endpoint inside $[0, 1/2)^{2k}$. As all points lie inside $[1/2, 1]^{2k}$, we can extend any large empty until its "lower left" corner is the origin. 

After these modifications, applying Lemma \ref{Lemma:LiftingLemma}, we can use the same arguments as in Section \ref{Sec:MaximumEmptyStar}: There is an empty box of volume $C^k$ if and only if $G$ has a $k$-clique. This proves the second part of Theorem \ref{Thm:EmptyBoxMain}.



\section{The Box Discrepancy Problem}\label{Sec:BoxDiscrepancy}
In order for our proof to work for this case, we will combine the ideas of the previous sections. Recall that we want to compute the box discrepancy
\[ \bar D_{\mathcal R}\left(P\right) := \max_{R \in \mathcal R} \left| vol(R) - \frac{|R \cap P|}{|P|} \right| \]
of a point set $P$.

To simplify our arguments, we make sure that any box containing \emph{all} points has volume $1$. Thereto, we add one point at the origin and one point at $(1, 1, \dots, 1)$. 

Now we construct the point set with the constants determined in Section \ref{Sec:StarDiscrepancy} (with $N$ increasing by $2$ because of the additional points). For this, we again choose $C$ large enough so that the maximum is attained for a large box with \emph{no points} inside, i.e., so that
\[ C^k > \left( \frac{N-1}{N} \right).\]
Finally, we lift all points (except for the origin) as in Section \ref{Sec:LargestEmptyBox}. Using the same arguments, Theorem \ref{Thm:ContinuousDiscrepancyMain} follows.


\section{Other geometric range spaces}
So far we have considered range spaces determined by (restricted or unrestricted) boxes inside the unit cube, namely $\calB$ and $\calB_0$. Similar questions can be asked when the ranges are determined by other (geometric) objects. We give a few examples and a short discussion on how to adapt the proofs to these ranges. We restrict ourselves to the analog of the $\textsc{Bichromatic-Rectangle}$ problem, from which the hardness of computing the combinatorial discrepancy is easily derived. That means for a point set $P = P_r \uplus P_b$ and range space $\calR$, we consider the value
\[ E_{\calR}(P_r, P_b) := \max_{R \in R, R \cap P_r = \emptyset} | R \cap P_b |. \]

\paragraph{Cubes}
If the range space is $\mathcal Q$, the set of all cubes inside the unit square, we can adapt the construction as shown in Figure \ref{subfig:CubeAdaption}.

\paragraph{Convex sets}
Here, the same arguments as in Section \ref{Sec:BichromaticRectangle} work as well: Any convex set that does not contain any red points can contain at most one blue point from each $\R_i^2$. Further, also as a direct consequence of convexity, the encoding of the edges works as well.

\paragraph{Half-spaces}
Instead of putting all points on a line, we now put all points on a convex curve as in Figure \ref{subfig:HalfSpaceAdaption}. Note the two additional red points on both ends to prevent $b_i(1)$ and $b_i(n)$ to be chosen at the same time (by a hyperplane that does not contain the origin).

\begin{figure}\label{fig:OtherRanges}
	\subfigure[Cubes]{\includegraphics[width=0.45\textwidth]{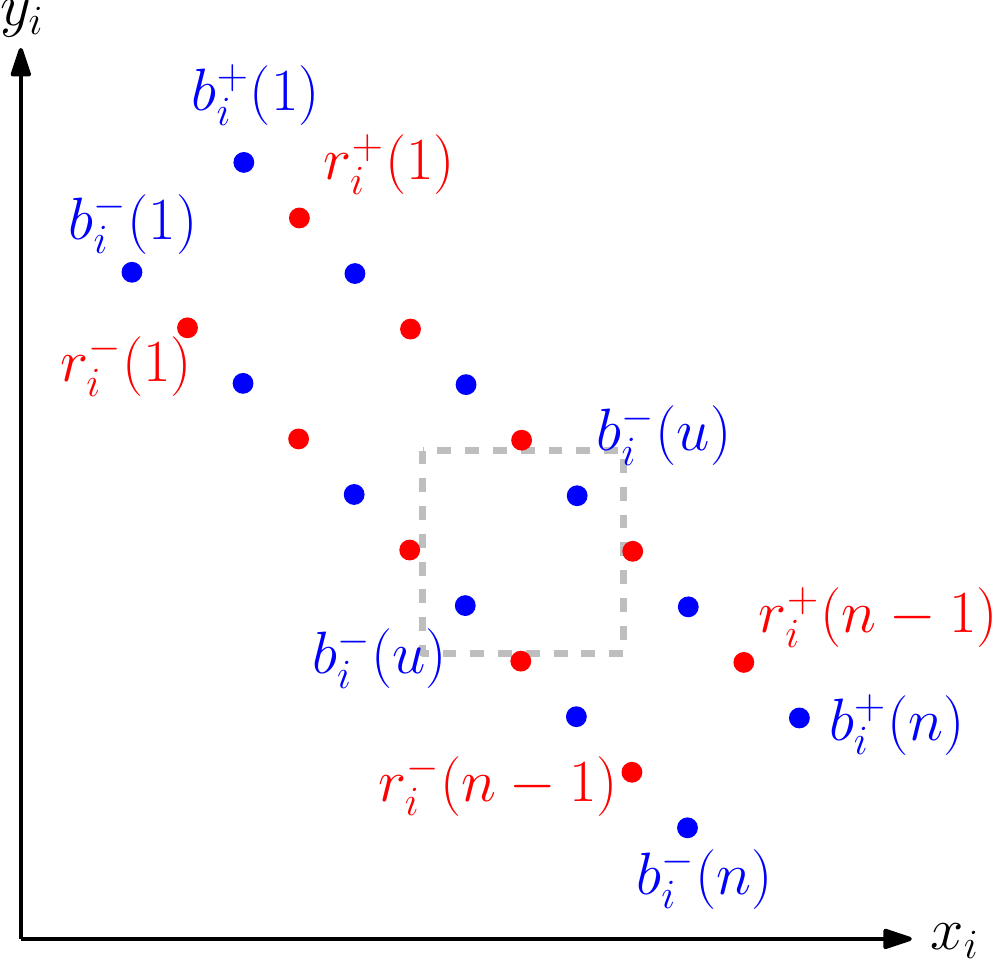}\label{subfig:CubeAdaption}}\hfill
	\subfigure[Half-spaces]{\includegraphics[width=0.45\textwidth]{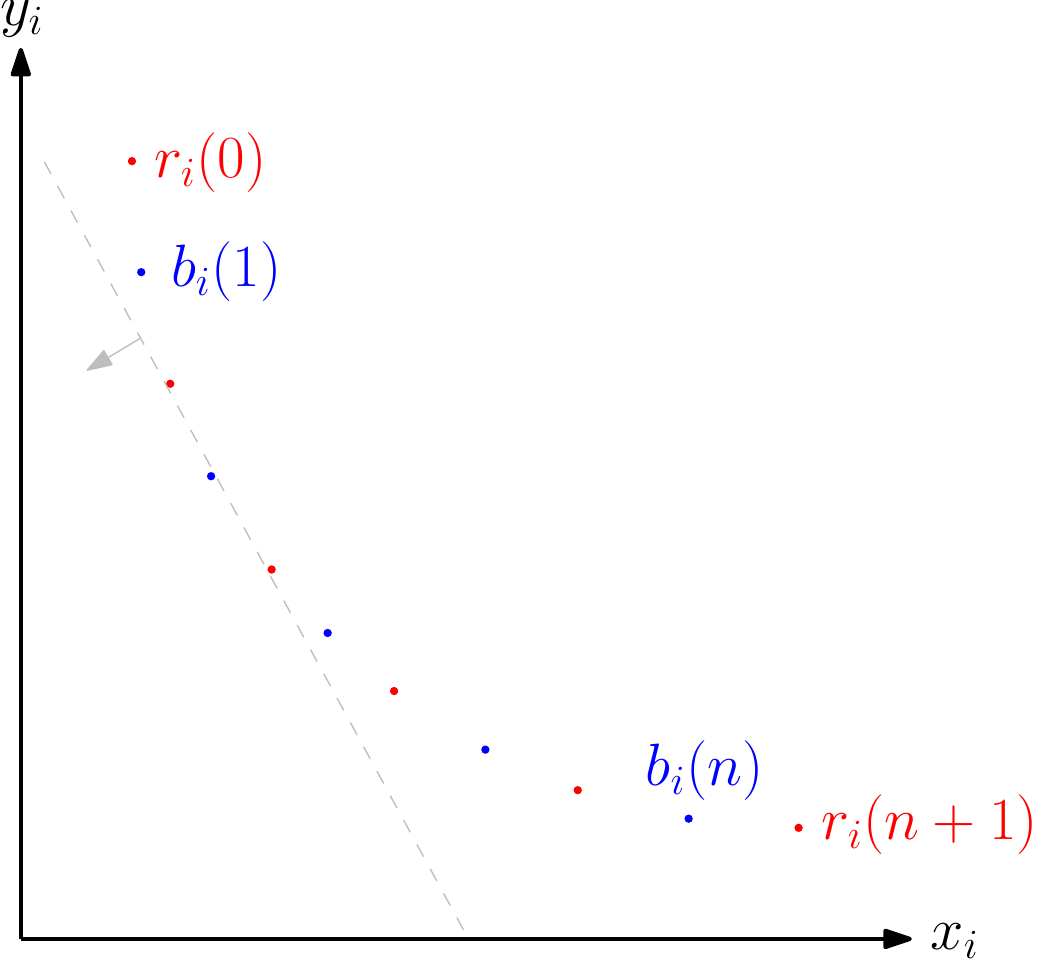}\label{subfig:HalfSpaceAdaption}}\hfill
	\caption{Other range spaces}
\end{figure}

\begin{theorem}\label{Thm:OtherRanges} The problems $d$-\textsc{Cube-Discrepancy}, $d$-\textsc{Convex-Sets-Discrepancy} and $d$-\textsc{Half-Space-Discrepancy} are W[1]-hard with respect to $d$.
\end{theorem}


\section{Implications on verification of epsilon-nets}
We concentrate on the problem $d$-\textsc{Net-Verification} for half-spaces. The proof can easily be adapted to the other range spaces considered.

Recall that, for given sets $S \subseteq P \subset \R^d$ and an $\eps > 0$, we want to decide whether every half-space that contains at least $\eps |P|$ points also contains a point from $S$.

To see why this problem is hard, consider the construction for the $d$-\textsc{Bichromatic-Rectangle} problem from Section \ref{Sec:BichromaticRectangle}. There, we had a set $P_r$ of red and a set $P_b$ of blue points. We have shown that it is W[1]-hard to decide whether there is a half-space containing $k$ blue and no red points. To reduce this problem to $d$-\textsc{Net-Verification}, we set $P$ to be the set of all points, and $S$ to be the set of red points $P_r$. Then, we set $\eps = k/|P|$. Now a  half-space containing $k$ points corresponds to a large set that is \emph{not} intersected: it contains $k = \eps|P|$ blue points but no red points. But this means that $S$ is \emph{not} an $\eps$-net. This proves Theorem \ref{Thm:EpsilonNetMain}.


\section{Conclusion}

As the problems we have considered are all computationally hard when $d$ is
part of the input, we have to resort to approximation algorithms when
dealing with them. We have seen that for the \textsc{Maximum-Empty-Star}
problem, there is no hope for a polynomial time algorithm with a reasonable
approximation factor at all.
The (in)approximability of the (star) discrepancy is open even in the
classical complexity theory framework. The question whether it can be
approximated in FPT time when parameterized by dimension seems worthwhile
considering.

\bibliographystyle{alpha}
\bibliography{Discrepancy}

\newcommand{\etalchar}[1]{$^{#1}$}
\begin{thebibliography}{CGK{\etalchar{+}}11}

\bibitem[AHP08]{AH08}
Boris Aronov and Sariel Har-Peled.
\newblock {On Approximating the Depth and Related Problems}.
\newblock {\em SIAM J. Comput.}, 38(3):899--921, 2008.

\bibitem[Alo11]{Al11}
Noga Alon.
\newblock {A Non-linear Lower Bound for Planar Epsilon-nets}.
\newblock {\em Discrete and Computational Geometry}, 2011.
\newblock To appear. DOI:10.1007/s00454-010-9323-7.

\bibitem[AS87]{AS87}
Alok Aggarwal and Subhash Suri.
\newblock {Fast Algorithms for Computing the Largest Empty Rectangle}.
\newblock In {\em Symposium on Computational Geometry}, pages 278--290, 1987.

\bibitem[BK09]{BK09}
Jonathan Backer and J.~Mark Keil.
\newblock {The Bichromatic Rectangle Problem in High Dimensions}.
\newblock In {\em Proc. of the 21st Annual Canadian Conference on Computational
  Geometry, Vancouver, BC, Canada}, pages 157--160, 2009.

\bibitem[BK10]{BK10}
Jonathan Backer and J.~Mark Keil.
\newblock {The Mono- and Bichromatic Empty Rectangle and Square Problems in All
  Dimensions}.
\newblock In {\em Proc. of the 9th Latin American Theoretical Informatics
  Symposium}, volume 6034 of {\em LNCS}, pages 14--25, 2010.

\bibitem[CCF{\etalchar{+}}05]{CCFHJKX05}
Jianer Chen, Benny Chor, Mike Fellows, Xiuzhen Huang, David Juedes, Iyad~A.
  Kanj, and Ge~Xia.
\newblock Tight lower bounds for certain parameterized {NP}-hard problems.
\newblock {\em Inf. Comput.}, 201:216--231, 2005.

\bibitem[CGK{\etalchar{+}}11]{cgkmr-gcfpt-09}
Sergio Cabello, Panos Giannopoulos, Christian Knauer, D\'aniel Marx, and
  G{\"u}nter Rote.
\newblock Geometric clustering: fixed-parameter tractability and lower bounds
  with respect to the dimension.
\newblock {\em ACM Transactions on Algorithms}, 2011.
\newblock To appear.

\bibitem[CGKR08]{CGKR08}
Sergio Cabello, Panos Giannopoulos, Christian Knauer, and G{\"u}nter Rote.
\newblock Geometric clustering: fixed-parameter tractability and lower bounds
  with respect to the dimension.
\newblock In {\em Proc. 19th Ann. ACM-SIAM Sympos. Discrete Algorithms}, pages
  836--843, 2008.

\bibitem[Cha00]{Ch00}
Bernard Chazelle.
\newblock {\em The discrepancy method: randomness and complexity}.
\newblock Cambridge University Press, 2000.

\bibitem[CT97]{DBLP:journals/ipl/CesatiT97}
Marco Cesati and Luca Trevisan.
\newblock {On the Efficiency of Polynomial Time Approximation Schemes}.
\newblock {\em Inf. Process. Lett.}, 64(4):165--171, 1997.

\bibitem[DEM96]{DEM96}
David~P. Dobkin, David Eppstein, and Don~P. Mitchell.
\newblock {Computing the Discrepancy with Applications to Supersampling
  Patterns}.
\newblock {\em ACM Trans. Graph.}, 15(4):354--376, 1996.

\bibitem[DF99]{DF99}
Rod~G. Downey and Michael~R. Fellows.
\newblock {\em Parameterized Complexity}.
\newblock Monographs in Computer Science. Springer, 1999.

\bibitem[DGW10]{DBLP:journals/jc/DoerrGW10}
Benjamin Doerr, Michael Gnewuch, and Magnus Wahlstr{\"o}m.
\newblock Algorithmic construction of low-discrepancy point sets via dependent
  randomized rounding.
\newblock {\em J. Complexity}, 26(5):490--507, 2010.

\bibitem[DJ09]{DJ09}
Adrian Dumitrescu and Minghui Jiang.
\newblock On the largest empty axis-parallel box amidst n points.
\newblock {\em CoRR}, abs/0909.3127, 2009.

\bibitem[DT97]{DT97}
Michael Drmota and Robert~F. Tichy.
\newblock {\em Sequences, Discrepancies and Applications}, volume 1651 of {\em
  Lecture Notes in Mathematics}.
\newblock Springer, 1997.

\bibitem[EHL{\etalchar{+}}02]{606903}
Jonathan Eckstein, Peter~L. Hammer, Ying Liu, Mikhail Nediak, and Bruno
  Simeone.
\newblock {The Maximum Box Problem and its Application to Data Analysis}.
\newblock {\em Comput. Optim. Appl.}, 23(3):285--298, 2002.

\bibitem[FG06]{FG06}
J{\"o}rg Flum and Martin Grohe.
\newblock {\em Parameterized Complexity Theory}.
\newblock Springer, 2006.

\bibitem[Gne08]{DBLP:journals/jc/Gnewuch08}
Michael Gnewuch.
\newblock Bracketing numbers for axis-parallel boxes and applications to
  geometric discrepancy.
\newblock {\em J. Complexity}, 24(2):154--172, 2008.

\bibitem[GSW09]{GSW09}
Michael Gnewuch, Anand Srivastav, and Carola Winzen.
\newblock Finding optimal volume subintervals with $k$ points and calculating
  the star discrepancy are {NP}-hard problems.
\newblock {\em J. Complexity}, 25:115--127, 2009.

\bibitem[HW86]{HW86}
David Haussler and Emo Welzl.
\newblock Epsilon-nets and simplex range queries.
\newblock In {\em Proceedings of the second annual symposium on Computational
  geometry}, SCG '86, pages 61--71. ACM, 1986.

\bibitem[Mat10]{Ma10}
Ji\v{r}\'{i} Matou\v{s}ek.
\newblock {\em {Geometric Discrepancy: An Illustrated Guide}}.
\newblock Springer, 2010.

\bibitem[Nie92]{niederreiterbook}
Harald Niederreiter.
\newblock {\em {Random number generation and quasi-Monte Carlo methods}}.
\newblock Society for Industrial and Applied Mathematics, Philadelphia, PA,
  USA, 1992.

\bibitem[Nie06]{Nie06}
Rolf Niedermeier.
\newblock {\em Invitation to Fixed-Parameter Algorithms}.
\newblock Oxford University Press, 2006.

\bibitem[PT10]{PT10}
J{\'a}nos Pach and G{\'a}bor Tardos.
\newblock Tight lower bounds for the size of epsilon-nets.
\newblock {\em CoRR}, abs/1012.1240, 2010.

\bibitem[Thi01]{DBLP:journals/jc/Thiemard01}
Eric Thi{\'e}mard.
\newblock {An Algorithm to Compute Bounds for the Star Discrepancy}.
\newblock {\em J. Complexity}, 17(4):850--880, 2001.

\bibitem[VC71]{VC71}
Vladimir~N. Vapnik and Alexey~Ya. Chervonenkis.
\newblock On the uniform convergence of relative frequencies of events to their
  probabilities.
\newblock {\em Theory of Probability and its Applications}, 16(2):264--280,
  1971.

\end{thebibliography}

\end{document}